%% file: main.tex
\documentclass{article}
\usepackage[a4paper,margin=1in]{geometry}
\usepackage{booktabs} 
\usepackage[ruled]{algorithm2e} 

\SetAlFnt{\small}
\SetAlCapFnt{\small}
\SetAlCapNameFnt{\small}
\SetAlCapHSkip{0pt}
\IncMargin{-\parindent}

\usepackage{amsmath, amssymb, amsthm}
\usepackage{mathtools}
\usepackage{enumerate}
\usepackage{nicefrac}
\usepackage{bbm}
\usepackage{bm}
\usepackage[dvipsnames]{xcolor}
\usepackage{hyperref}
\hypersetup{
     colorlinks   = true,
     linkcolor    = MidnightBlue,
     citecolor    = MidnightBlue,
     urlcolor     = MidnightBlue
}
\usepackage[capitalize]{cleveref}

\usepackage[suppress]{color-edits}
\definecolor{mygreen}{RGB}{20,180,20}
\addauthor{et}{blue}
\addauthor{er}{purple}
\addauthor{jk}{mygreen}
\usepackage{dsfont}

\newtheorem{thm}{Theorem}
\newtheorem{defn}[thm]{Definition}
\newtheorem{lem}[thm]{Lemma}
\newtheorem{cor}[thm]{Corollary}

\newtheorem{prop}[thm]{Proposition}

\title{Ordered Submodularity and its Applications to Diversifying Recommendations}

\author{
    Jon Kleinberg\thanks{Supported in part by a Vannevar Bush Faculty Fellowship, MURI grant W911NF-19-0217, AFOSR grant FA9550-19-1-0183, a Simons Collaboration grant, and a grant from the MacArthur Foundation.}\\
    Cornell University\\
    \texttt{kleinberg@cornell.edu}
    \and
    Emily Ryu\thanks{Supported in part by AFOSR grant FA9550-19-1-0183.}\\
    Cornell University\\
    \texttt{eryu@cs.cornell.edu}
    \and
    \'Eva Tardos\thanks{Supported in part by NSF grant CCF-1408673 and AFOSR grant FA9550-19-1-0183.}\\
    Cornell University\\
    \texttt{eva.tardos@cornell.edu}
}

\begin{document}

\maketitle

\begin{abstract}
    \input{sections/0_abstract}
\end{abstract}

\input{sections/1_intro}

\input{sections/2_relatedwork}
\input{sections/3_defn_basiclemmas}
\input{sections/4_greedy}
\input{sections/5_coverage}
\input{sections/6_calibration}
\input{sections/7_discussion}

\bibliographystyle{alpha}
\bibliography{main.bib}

\newpage
\appendix
\input{sections/appendix}

\end{document}

%% file: sections/0_abstract.tex

A fundamental task underlying many important optimization problems, from influence maximization to sensor placement to content recommendation, is to select the optimal group of $k$ items from a larger set. Submodularity has been very effective in allowing approximation algorithms for such subset selection problems. However, in several applications, we are interested not only in the elements of a set, but also the \emph{order} in which they appear, breaking the assumption that all selected items receive equal consideration. One such category of applications involves the presentation of search results, product recommendations, news articles, and other content, due to the well-documented phenomenon that humans pay greater attention to higher-ranked items. As a result, optimization in content presentation for diversity, user coverage, calibration, or other objectives more accurately represents a sequence selection problem, to which traditional submodularity approximation results no longer apply. Although extensions of submodularity to sequences have been proposed, none is designed to model settings where items contribute based on their position in a ranked list, and hence they are not able to express these types of optimization problems. In this paper, we aim to address this modeling gap.

Here, we propose a new formalism of \emph{ordered submodularity} that captures these ordering problems in content presentation, and more generally a category of optimization problems over ranked sequences in which different list positions contribute differently to the objective function. We analyze the natural ordered analogue of the greedy algorithm and show that it provides a $2$-approximation. We also show that this bound is tight, establishing that our new framework is conceptually and quantitatively distinct from previous formalisms of set and sequence submodularity.


%% file: sections/1_intro.tex


\section{Introduction} \label{sec:intro}

Many important optimization problems involve selecting a subset
of items from a larger set. Examples of such tasks include influence
maximization in social networks~\cite{KKT03}, sensor placement and
experimental design~\cite{krause2008robust}, and recommendation
systems~\cite{Yue11, gabillon2013adaptive}. In domains in which the
goal is this type of subset selection,
\emph{submodularity} has been widely used to express the notion of
``diminishing marginal returns.'' Submodularity is a powerful
framework for approximate optimization; in particular, there is a rich
literature on approximation algorithms for selecting subsets 
achieving near-maximum value with respect to a submodular function
\cite{nemhauserwolseyfisher1978,
calinescu2011, filmus2013, vondrak2013, krause2014survey}.

An implicit modeling assumption in the use of submodularity is
that the order of the selected elements does not matter;
this is crucial, since submodularity is a property of functions
that operate on unordered sets.
However, in many applications, we are interested not only in the
elements of a set, but also the order in which the elements appear.
A broad category of such applications, in both on-line and off-line
settings, is the presentation of content to an audience ---
for example, search results, product or entertainment recommendations, 
news articles, social media posts, and many other instances.
Content presentation crucially depends on sequential effects 
due to well-documented phenomena in human behavior ---
specifically, that human cognition is generally limited to
serially processing information one piece at a time, rather than
processing all elements of a list in parallel. Moreover, people tend
to have limited attention span and patience, meaning that when
items of content are presented in a ranked list,
the higher-ranked items are likely to receive significantly greater attention
\cite{pan2007google}. This results in several empirical
observations, such as inverse power law relationships in number of
clicks on search results~\cite{williamsZipf} and sharp decreases in
webpage viewing time ``below the fold'' (content that does not fit on
the first screen and must be scrolled down to reach)~\cite{NNGfold}.

The use of optimization frameworks for content presentation
suggests some of the fundamental limits in the application of submodularity
for problem domains where sequential effects are important.
In particular, for a number of basic problems in ranking and recommendation,
standard formalisms model them as the problem of selecting
a subset of items to present to a user, then showing that
the resulting objective function over selected subsets is
submodular, and thus deriving guarantees for
approximating this objective function.
But if the value of a set of items to a user is strongly
dependent on the order in which it is presented, 
then the optimization is in fact taking place over sequences rather
than sets, and in this richer formalism submodularity would not
be applicable.

Our goal in this paper is to propose a formalism that can address 
these types of ordering issues in optimization problems generally,
and for a collection of basic content presentation problems in particular.
We begin by observing (and demonstrating in Section~\ref{sec:relatedwork})
that while other generalizations of submodularity to sequences 
have been formulated, they fundamentally make assumptions that are not
well-suited to modeling the sequential effects that arise from
phenomena like the diminshing attention of a user reading a ranked list.
Hence, a new notion of
submodularity for sequences is required. Here, we present such a
generalization of a combined monotonicity-submodularity property,
which we term \emph{ordered submodularity}.
We provide approximation guarantees for functions of this type,
and we show how they capture the sequential effects in a
range of standard content presentation problems.

\paragraph{\bf Motivating applications.} 
Throughout our work, it is useful to keep in mind the following two
standard problems in ranking and recommendation that help motivate our work.
The first is a {\em coverage} problem that is used for creating diversity
in ranked lists of items as follows 
\cite{agrawal2009, ashkan2015}.
Suppose we want to produce a list of $k$ recommendations (say of movies)
to show to a group of users.
Each movie can satisfy only some subset of the users, and we would like
to choose the $k$ movies so to maximize the number of users 
who like at least one item on the list.
(In this way, we seek to cover their preferences as completely as
possible with $k$ items.)
We can view the number of users satisfied 
as an objective function on the set of
$k$ items chosen; in \cite{agrawal2009} it is shown that this function
is monotone and submodular, and hence greedy maximization provides a
$(1-\frac{1}{e})$-approximation.
But as the authors of \cite{agrawal2009} observe, in the real application
users will have declining attention as they process the list of items,
and different users will stop reading the list at different points.
This basic addition to the model --- that users have
differential patience --- 
means that the order of the list is crucial for evaluating
the number of users that it satisfies; and once we introduce
ordering into the problem, the results from the large body of work on
submodular optimization no longer hold in this setting.
Is there still a way to find good approximations to the optimal 
ranked list?

The second problem we draw on for motivation is the task of
{\em calibrating} recommendations \cite{Steck18}.
In this problem, we present a list of $k$ recommendations
to a single user (again, suppose they are movies);
and we assume that each movie represents a distribution over {\em genres}.
(For example, a documentary in Italian about the national soccer team is a multi-genre mixture of a movie about sports, an Italian language film, and a documentary.)
The list of $k$ items thus induces an {\em average} distribution
over genres.
Now, the user has a {\em target distribution} over genres that
reflect the extent to which they want to consume each genre
in the long run.
A natural goal is that the average distribution induced by the
list of recommendations should be ``close'' (in a distributional sense)
to the target genre distribution of the user;
when these two distributions are close, we say that the set of
recommendations is {\em calibrated} to the user.
(For example, a user who likes both Italian language films and movies about sports might well be dissatisfied with recommendations consisting only of sports movies in English;
this set of recommendations would be badly calibrated to the
user's target distribution of genres.) 
For natural measures of distributional similarity, the selection of a
set of $k$ items to match the user's target distribution can be
formulated as the maximization of a submodular set function.
But here too, the work introducing this problem observed that since
user attention diminshes over the course of a ranked list, the
list of $k$ items is really producing a {\em weighted average} over
the genres of these items, with the earlier items in the list 
weighted more highly than the later ones \cite{Steck18}.
Once we introduce this natural addition to the problem, based on
ordering, it again
becomes unclear whether there are good algorithms to find provably
well-calibrated lists of recommendations.

\paragraph{\bf A new definition of ordered submodularity.}
In this paper we introduce a property called {\em ordered submodularity}
that can be viewed as an analogue of monotonicity and submodularity
for functions defined on sequences.
It captures both of the motivating applications described above,
and more generally captures a category of optimization problems 
which search over lists, and in which different list positions 
contribute differently to the objective function.

We define the property as follows.
Let $f$ be a function defined on a sequences of elements from some ground set;
we say that $f$ \emph{ordered-submodular} if for all
sequences of elements $s_1 s_2 \dots s_k$, the following property
holds for all $i \in [k]$ and all other elements $\bar{s}_i$: 
$$f(s_1 \dots s_i) -
f(s_1 \dots s_{i-1}) \ge f(s_1 \dots s_i \dots s_k) - f(s_1 \dots
s_{i-1} \bar{s}_i s_{i+1} \dots s_k).$$

Notice that if $f$ is an ordered-submodular function that takes
sequences as input but does not depend on their order
(that is, it produces the same value for all permutations of
a given sequence), then it follows immediately from the definition
that $f$ is a monotone submodular set function.
In this way, monotone submodular set functions are a special
case of our class of functions.

We prove that for any ordered-submodular function $f$,
the natural greedy algorithm for maximizing $f$ --- building 
a sequence by always appending the item that produces the
largest marginal gain --- is a $2$-approximation, and there
are simple examples of ordered-submodular functions for which
the greedy algorithm does no better than a factor of $2$.
This highlights a key distinction from the unordered case of
monotone submodular set functions: there the corresponding
greedy algorithm produces the strictly better approximation
guarantee of $(1 - 1/e)$. 
Hence the move to ordered submodularity
changes the approximability of the maximization problem
in a qualitative way: it still admits a small constant-factor bound,
but a different constant.

In the coverage problem described above 
with users of differential patience, we show directly that
the objective function is ordered-submodular, and this provides
the first non-trivial approximation guarantee for this problem.
(This problem provides some of the simple examples
in which the factor of 2 is tight for the performance of
the greedy algorithm.)
For the calibrated recommendation problem with ranked lists
described above, we need to specify how the distance between
distributions will be measured;
we show that that natural ways of measuring distance (such as the
classical family of $f$-divergences from the statistics and
information theory literature) give rise to ordered-submodular functions.
We thus obtain the first non-trivial approximation guarantee for this 
ordered problem as well.
As noted above, we find it interesting that existing formalisms
extending submodularity to sequences do not capture the
objective functions arising from problems such as these two,
and the way in which items in these problems 
contribute based on their position in a ranked list.\footnote{As one
indication of the differences at work, these earlier formalisms 
for submodularity over sequences have the property that 
the greedy algorithm continues to be a $(1-1/e)$-approximation
for the corresponding maximization problem.
But for the ordered coverage problem we have described here,
the greedy algorithm can differ from the optimum by a factor of 2;
this is the tight bound on its approximation performance, and it suggests
that the problem has a qualitatively different type of objective function.}
In the next section, we provide some detail for why these
alternative formalisms differ from our proposal and do not capture the objective functions
we consider in the paper; following this, we establish our
approximation results and their application to the problems discussed here.

%% file: sections/2_relatedwork.tex
\section{Related work} \label{sec:relatedwork}

First, we cover general theories of submodularity in sequences and explain how they cannot model the types of problems that our definition does. Then, we discuss applications in the specific context of recommender systems.

\subsection{Existing frameworks for submodularity in sequences} 
Alaei, Makhdoumi, and Malekian (2010) introduce the first generalizations of \emph{sequence-submodularity} and \emph{sequence-monotonicity} in the context of online ad allocation, and show that the greedy algorithm for sequence-submodular maximization achieves a $(1-\frac{1}{e})$-approximation to the optimal solution~\cite{Alaei19}. However, a major limitation of their model is that their definition of sequence-monotonicity is extremely strong. Their result requires that $f(A) \le f(B)$ for any sequences $A$ and $B$ such that $A$ is a subsequence of $B$, which in many settings is too restrictive to be useful. For instance, if an element $s_1$ only contributes to the value of the objective function when included as the first element of the input sequence but not as the second, it is possible to have $f(s_1) > f(s_2 s_1)$, violating sequence-monotonicity.

Similarly, Zhang et al. (2013) study the maximization of \emph{string submodular} functions of strings (or sequences) of actions chosen from a set, a notion similar to sequence-submodularity but only requiring monotonicity and diminishing returns with respect to prefixes, not all subsequences~\cite{Zhang16}. When the function also satisfies monotonicity with respect to postfixes and not only prefixes, then they, too, establish a $(1-\frac{1}{e})$-approximation ratio for the greedy algorithm, and provide improved guarantees when additional curvature constraints are satisfied. More formally, \emph{prefix/postfix monotonicity} requires that for any sequences $A$ and $B$ and their concatenation $A||B$, it must hold that $f(A||B)\ge f(A)$ (prefix monotonicity) \emph{and} $f(A||B) \ge f(B)$ (postfix monotonicity), properties which were both previously suggested by Streeter and Golovin (2008), who considered sequences in the context of an online submodular selection problem~\cite{StreeterGolovin}. As seen above, postfix monotonicity is not a natural property when modeling attention drop-off, since it would imply that prepending a ``bad'' movie that interests nobody at the front of a ranked list would capture more users, which clearly is not the case.

In another direction, Tschiatschek, Singla, and Krause (2017) approach the selection of maximizing sequences using submodularity by encoding sequential dependencies in a directed acyclic graph~\cite{Tschiatschek2017}, and Mitrovic et al. (2018) generalize this concept from DAGs to hypergraphs~\cite{Mitrovic18}. They place an edge between two nodes $(u,v)$ of the graph if there is additional utility in selecting element $u$ before element $v$, and then consider submodular functions on the edge set of the graph. However, this approach is only able to represent sequential dependencies inherent to the identity of a set of elements (for example, watching a prequel before the sequel), but it cannot represent decreasing attention or other complex dependencies dependencies that may vary with the objective function, or with the position and identity of other elements in the input sequence. 

Most recently, Bernardini, Fagnani, and Piacentini (2021) propose a framework in which the set of all elements is equipped with some property $g$, according to which it has a total ordering. Their objective function is defined recursively as the sum of the marginal increase of appending each element $\sigma$ to the list of earlier elements, weighted by $g(\sigma)$. Denoting the subsequence of the first $i$ elements in the list as $S_i$, for any function $g$ and any monotone submodular set function $h$, they study sequence functions of the form 
$$f(s_1 \dots s_k) = \sum_{i=1}^k g(s_i) \cdot [ h(S_i) - h(S_{i-1}) ].$$
Phrased this way, the sequential nature of the problem results from considering the marginal increase due to each element with respect to the set of elements before it, but the weight assigned to each marginal increase depends solely on the \emph{identity} of the element, not its rank. While this is a valid assumption in a number of applications, it does not hold in our particular use case of modeling sequential attention drop-off. In contrast, our framework encompasses functions of the form 
$$f(s_1 \dots s_k) = \sum_{i=1}^k g_i \cdot [ h(S_i) - h(S_{i-1}) ],$$ 
where $g_i$ can be thought of as the weight assigned to \emph{rank} $i$. This key difference allows us to avoid imposing a total $g$-ordering on the set of all elements (even if such an ordering does exist, this information may not be known to a system designer). Perhaps more significantly, it also introduces an additional sequential aspect that further differentiates our approach from traditional set submodularity.

\subsection{Applications to diversifying and calibrating recommendations} 

One important topic in content presentation is the problem of curating search results that are useful to a diverse population of users. Agrawal et al. (2009) establish a mathematical formalization of this \emph{user coverage} problem, which they study through the lens of submodularity~\cite{agrawal2009}. They suppose that each item has some probability of satisfying every user type. Then, they seek to display a diverse set of search results to maximize the number of users who find at least one satisfactory document. That is, given a query, they seek to select $k$ search results to maximize the probability that a randomly chosen user drawn from a heterogeneous group likes at least one item in the set. The authors show that this objective is a monotone submodular set function, and consequently observe that there exists a $(1-\frac{1}{e})$-approximation algorithm for the problem. While mathematically elegant, a key limitation of this formulation is that it assumes that all users are equally patient and give equal consideration to all search results. Acknowledging that this is not an accurate representation of human patience and attention in the real world, the authors suggest as a direction for future work the formulation of an objective function that accounts for the distribution of users who stop at different points in the search results. Our work does exactly this. In doing so, the presentation order of the search results becomes important, and the objective function becomes a sequence function that must be studied using our new definition of ordered submodularity. 

Ashkan et al. (2015) also study diversification for user coverage in recommender systems, this time using a modular function subject to a submodular constraint~\cite{ashkan2015}. They maintain the consideration that recommendations should not be \emph{only} diverse, but also still broadly relevant and useful, by maximizing a weighted sum of a diversification metric and the sum of all the utilities of the recommended items. In their setup, the greedy approach to maximization is optimal. But again, their formulation assumes that all users have equal patience and consider all recommendations equally, so their optimality result does not hold when users have differential patience values. 



Steck (2018) also considers the question of creating diversity in lists of recommendations, but with the different goal of creating recommendations that are \emph{calibrated} to the user's interests~\cite{Steck18}. (We note that in the literature, ``diversification'' has historically been used to refer to variants of the coverage problem previously discussed, but we find it more useful to think of ``diversity'' as a general concept describing lists that include a mixture of categories. The coverage objective is one way to achieve diversity by including as heterogeneous a mixture as possible; the calibration objective is another way that includes categories in a proportional mixture. Section 5.1 of~\cite{Steck18} discusses the relationship between diversity, calibration, and other metrics in more detail.) Steck proposes as a heuristic for calibration a modified version of the KL divergence from the recommended distribution to the user's preference distribution. When all the recommended items are assigned equal weight, this induces a submodular set function, which can be used for approximate maximization via the greedy algorithm. But in the case when the recommended items have unequal weights, such as when accounting for attention dropoff, his approximation results do not apply. We discuss more of the technical details of Steck's formalism in Section~\ref{sec:calibration} and describe extensions to our ordered-submodular optimization framework for sequences.

Lastly, another setting in which some notion of weights appears in submodular optimization is the context of knapsack constraints or budgets~\cite{Sviridenko, Alon2012, SomaKIK14}. Here, we note that despite the initial similarities in terminology, the use of weights as capacities in this line of work is quite different from the attenuation of attention and impact that we intend our weights to represent.


%% file: sections/3_defn_basiclemmas.tex


\section{Definition of ordered submodularity}\label{sec:ordsub_defn}

In this section we define our extension of submodularity to ordered sets. To simplify notation, for two sequences $A$ and $B$ we will use $A||B$ to denote their concatenation. For a single element $s$ we will use $A||s$ to denote $s$ added at the end of the list $A$.

\begin{defn}[Ordered submodularity]
A sequence function $f$ is \textbf{ordered-submodular} if for all sequences $A$ and $B$, the following property holds for all elements $s$ and $\bar s$: $$f(A||s)-f(A)\ge f(A|| s||B)-f(A||\bar s||B).$$
\end{defn}

Ordered submodularity can be viewed as a generalization of monotonicity and submodularity for set functions. For functions $f$ that depend only on the set of elements in the input sequence and not their order, setting $\bar s=s$ implies $f(A||s)-f(A)\ge 0$, corresponding to monotonicity, and setting $\bar s$ to the ``null'' element implies 
\begin{align*}
    &f(A||s)-f(A) \ge f(A||s||B)-f(A||B),
\end{align*}
corresponding to submodularity.

On the other hand, any monotone submodular set function $f$ when viewed as a function on sequences, that does not depend on the order of the elements satisfies
\begin{align*}
f(A||s)-f(A)\ge f(A||s||B)-f(A||B)\ge f(A||\bar s||B)
\end{align*}
where the first inequality is due to submodularity and the second inequality is due to monotonicity. This is exactly ordered submodularity when $f$ is interpreted as a sequence function, so we see that ordered submodularity is indeed a very natural and well-motivated property in the sequential setting.

We now demonstrate a few basic ways of constructing ordered-submodular functions from other submodular and ordered-submodular functions.

\begin{lem} \label{lem:linear_combo}
If $f$ and $g$ are ordered-submodular, then $\alpha f+ \beta g$ is also ordered-submodular for any $\alpha, \beta \ge 0$.
\end{lem}
\begin{proof}
We simply multiply and add the two inequalities from the definition of ordered submodularity:
\begin{align*}
    \alpha \left[f(A || s) - f(A)\right] &\ge \alpha \left[f(A || s || B) - f(A || \bar{s} || B) \right] \\
    \beta \left[g(A || s) - g(A)\right] &\ge \beta \left[g(A || s || B) - g(A || \bar{s} || B) \right] \\
    \implies (\alpha f+ \beta g)(A||s) - (\alpha f+ \beta g)(A) &\ge (\alpha f+ \beta g)f(A || s || B) - (\alpha f+ \beta g)f(A || \bar{s}_i || B).
\end{align*}
\end{proof}

\begin{lem} \label{lem:threshold_submodular}
Suppose $h$ is a monotone submodular set function. Then the function $f$ constructed by evaluating $h$ on the set of the first $t$ elements of $S$, that is, 
$$f(S) = \begin{cases} h(S) & \text{if } |S| \le t \\
h(S_t) &\text{if } |S| > t \end{cases}$$
is ordered-submodular.
\end{lem}
Here, it is useful to think of $t$ as a \emph{threshold} beyond which additional elements contribute nothing to the value of $f$. Once again, $S_i$ denotes the sequence of the first $i$ elements of the sequence $S$, and for a sequence $S$ we use $h(S)$ to denote the value of the submodular function on the set of elements in $S$, independent of order.

\begin{proof}
We seek to show that for all sequences $A$ and $B$ and elements $s$ and $\bar{s}$, $$f(A||s)-f(A)\ge f(A|| s||B)-f(A||\bar s||B).$$ We take two cases based on $|A|$.

\paragraph{Case 1:} $|A| \ge t$.
Then $f(A || s) = f(A) = f(A||s||B) = f(A||\bar{s}||B) = h(A_t)$, so $$f(A||s) - f(A) = 0 = f(A||s||B) - f(A||\bar{s}||B).$$

\paragraph{Case 2:} $|A| < t$.

Let $j = t - |A| - 1$. Now, observe that we have
\begin{align*}
    f(A||s) - f(A) = h(A||s) - h(A) &\ge h(A||s||B_j) - h(A||B_j) \\
    &\ge h(A||s||B_j) - h(A||\bar{s}||B_j) = f(A||s||B) - f(A||\bar{s}||B),
\end{align*}
where the first inequality is due to submodularity of $h$ and the second is due to monotonicity of $h$.
\end{proof}

\begin{lem} \label{lem:weighted_submodular}
Suppose $h$ is a monotone submodular set function and $\{g_i\}$ is a sequence of monotonically decreasing weights (i.e., $g_i \ge g_j$ if $i < j$). Then the sequence function defined by
$$f(S) = \sum_{i=1}^k g_i \cdot [ h(S_i) - h(S_{i-1}) ],$$ 
where $k = |S|$, is ordered-submodular. 
\end{lem}
Here we use $S_i$ to denote the sequence of the first $i$ elements of the sequence $S$, and for a sequence $S$ we use $h(S)$ to denote the value of the submodular function on the set of elements in $S$, independent of the order of the sequence.
\begin{proof}
Define $g_i' = g_i - g_{i+1}$ (where we use an additional term, $g_{k+1} = 0$, for notational convenience) and the sequence functions $h_i(S) = h(S_i)$, so that we can write $f(S) = \sum_{i=1}^k g_i' \cdot h_i(S)$. By monotonicity, $g_i' \ge 0$ for all $i$, so by Lemma~\ref{lem:linear_combo} it suffices to show that each $h_i(S)$ is ordered-submodular. But $h_i(S)$ is just a monotone submodular set function $h$ evaluated on a threshold of the first $i$ elements of $S$, so it is ordered-submodular by Lemma~\ref{lem:threshold_submodular}. Thus we conclude that $f$ is ordered-submodular.
\end{proof}


%% file: sections/4_greedy.tex


\section{Analysis of simple greedy algorithm}\label{sec:greedy_analysis}

The simple greedy algorithm for cardinality-constrained nonnegative ordered-submodular maximization works as follows: It initializes $A_0 = \emptyset$ (the empty sequence), and for $\ell = 1, 2, \dots, k$, it selects $A_\ell$ to be the sequence that maximizes $f(A)$ over all sequences obtained by appending an element to the end of $A_{\ell-1}$. In other words, it iteratively appends elements to the sequence $A$ one by one, each time choosing the element that leads to the greatest marginal increase in the value of $f$.

\begin{prop}\label{prop:2approx}
   The greedy algorithm for nonnegative ordered-submodular function maximization over sets of cardinality $k$ outputs a solution whose value is at least $\frac{1}{2}$ times that of the optimum solution.
\end{prop}
\begin{proof}
    Denote the sequence of length $k$ maximizing $f$ as $S = s_1 s_2 \dots s_k$ and the sequence of length $k$ maximizing the marginal increase at each step as $A = a_1 a_2 \dots a_k$. We write $S^j = s_j s_{j+1} \dots s_k$ to denote the suffix of $S$ starting at element $s_j$.
    
    Let $OPT(k) = f(S)$, $ALG(k) = f(A)$, so that we seek to show that $ALG(k) \ge \frac{1}{2} OPT(k)$ for all $k$. We must bound the performance of the greedy algorithm by comparing it to the optimal solution. The key insight is to ask the following question at each step: if we must remain committed to all the greedily chosen elements so far, but make the same choices as the optimum for the rest of the elements, how much have we lost? 
    
    To answer this question, we show via induction that for all $i$, $$f(A_i || S^{i+1}) \ge OPT(k) - f(A_i).$$ 
    
    The base case of $i=0$ is trivial, as $f(A_0 || S^1) = f(S) = OPT(k) \ge OPT(k) - f(A_0)$. So suppose the claim is true for some $i$, and observe that by ordered submodularity we have 
    \begin{align*}
        f(A_i || s_{i+1}) - f(A_i) &\ge f(A_i || s_{i+1} || S^{i+2}) - f(A_i || a_{i+1} || S^{i+2}) \\
        &= f(A_i || S^{i+1}) - f(A_{i+1} || S^{i+2}), \\
        f(A_{i+1} || S^{i+2}) &\ge f(A_i || S^{i+1}) + f(A_i) - f(A_i || s_{i+1}).
    \end{align*}
    
    Applying first the induction hypothesis, then the fact that $f(A_{i+1}) \ge f(A_i || s_{i+1})$ by definition of the greedy algorithm, yields
    \begin{align*}
        f(A_{i+1} || S^{i+2}) &\ge (OPT(k) - f(A_i)) + f(A_i) - f(A_i || s_{i+1}) \\
        &= OPT(k) - f(A_i || s_{i+1}) \\
        &\ge OPT(k) - f(A_{i+1}),
    \end{align*}
    completing the induction.
    
    Finally, taking $i=k$ in the claim gives 
    $$f(A) \ge OPT(k) - f(A) \implies f(A) = ALG(k) \ge \frac{1}{2} OPT(k).$$
\end{proof}


%% file: sections/5_coverage.tex


\section{Application 1: Diversification for user coverage}\label{sec:coverage}

Suppose we are designing a movie recommender system which produces a single list of recommendations for a large number of users. Every user has some amount of \emph{patience}, representing the fact that users are only willing to scroll down so far before deciding that the list is unsatisfactory. We say that the system \emph{covers} a user if the user is able to find a movie that interests them before their patience expires; otherwise the user gives up on the system and simply walks away. The goal of the designer is to diversify the list of recommendations in order to maximize the number of users covered by the system. In this section, we formally define an objective function for this problem and show that it is ordered-submodular, allowing us to conclude that the greedy algorithm gives a factor of $2$ approximation for the coverage problem.

\subsection{Mathematical formulation}
In a realistic application, we may not expect to exactly know each individual user that will ever use the recommendation system; instead, we may only know a probability distribution over the types of users who will use the system. We may also not know with complete certainty that a movie will or will not interest a given user; we may only have an estimated probability that a movie interests a user of a certain type. To generalize our model to this randomized setting, we seek to maximize the expected number of users covered by the system, or equivalently, the probability that a randomly chosen user is covered.

Let $\pi$ represent the probability distribution over user types (so that $\pi_u$ is the probability that a random user has type $u$). Denote the probability that movie $m$ interests user type $u$ by $p_{m,u}$. Define $\theta_u$, the \emph{patience} of type $u$, as the number of recommendations that a user of type $u$ will consider before leaving the system (e.g., if $\theta_u = 2$, the system will cover $u$ only if they are interested by the first or second movie in the list). 
Then, the probability that the recommendation list $S=s_1 s_2\dots s_k$ covers a randomly chosen user from $\pi$ is $$f(S) = \sum_u \pi_u \left(1 - \prod_{j=1}^{\min \{\theta_u, |S|\}} (1 - p_{s_j, u}) \right),$$ where the inner expression is obtained as the complement of the probability that a user of type $u$ is \emph{not} satisfied before their patience expires or they reach the end of the list, whichever comes first. This is the objective function that we now seek to maximize.

\subsection{Demonstration of ordered submodularity}
The objective function is of the form $f(S) = \sum_u \pi_u f_u(S)$, where $$f_u(S) = 1 - \prod_{j=1}^{\min \{\theta_u, |S|\}} (1 - p_{s_j, u}).$$ Thus by Lemma~\ref{lem:linear_combo}, to show ordered submodularity, it suffices to fix $u$ and show that $f_u$ is ordered-submodular. But now observe that $f_u$ is a function of the set of the first $\theta_u$ elements only (since multiplication is commutative, and any elements indexed above $\theta_u$ are not included in the product). Further, the coverage expression on the right hand side is a submodular set function of the type studied by~\cite{agrawal2009}. So $f_u$ is a sequence function defined by imposing a threshold $\theta_u$ on a submodular set function $h$, which is ordered-submodular by Lemma~\ref{lem:threshold_submodular}. Therefore, we conclude that the overall function $f$ is ordered-submodular. 

\begin{thm}
The user coverage function parametrized by user probability distribution $\pi$, movie satisfaction probabilities $\{p_{m,u}\}$, and patience values $\{\theta_u\}$, $$f(S) = \sum_u \pi_u \left(1 - \prod_{j=1}^{\min \{\theta_u, |S|\}} (1 - p_{s_j, u}) \right),$$ is ordered-submodular. Thus, the greedy algorithm produces a ranked list covering at least $\frac{1}{2}$ as many users as the optimal ranked list.
\end{thm}

\subsection{Greedy approximation ratio of 2 is tight}
A simple example in this setting shows that we can do no better than a factor of $2$ approximation using the greedy algorithm.

Suppose there are two user types, $1$ and $2$, with $(\pi_1, \pi_2) = \left(\frac{1}{2}, \frac{1}{2}\right)$, $\theta_1 = 1$, and $\theta_2 = 2$. There are also two movies, $s_1$ and $s_2$, with $p_{s_1, 1} = p_{s_2, 2} = 1$, $p_{s_1, 2} = p_{s_2, 1} = 0$. We seek to generate a recommendation list of length $2$ (i.e., to rank the two movies in order).

Since $f(s_1) = \pi_1 = \frac{1}{2}$, $f(s_2) = \pi_2 = \frac{1}{2}$, the greedy algorithm may choose arbitrarily between $s_1$ and $s_2$; suppose it chooses $s_2$ first.\footnote{We may also perturb the probabilities by an arbitrarily small amount $\varepsilon$ so that $\pi_1 < \pi_2$, but we make the standard assumption of arbitrary tiebreaking for a cleaner proof of the same result.} It then chooses $s_1$ in the second step, but obtains no additional value since $s_1$ only interests user type $1$, but user type $1$ will not look at the second movie in the list. Then $$ALG = f(s_2 s_1) = \pi_1 \cdot 0 + \pi_2 \cdot 1 = \frac{1}{2}.$$
But the optimal list would place $s_1$ ahead of $s_2$, which first covers user type $1$ before their patience expires, then covers user type $2$, giving 
\begin{align*}
    OPT = f(s_1 s_2) &= \pi_1 \cdot 1 + \pi_2 \cdot 1 = 1,
\end{align*}
so $ALG/OPT = \frac{1}{2}$ exactly.

This example can be extended to a recommendation list of arbitrary length $k$ by defining $k$ user types with $\pi_i = \frac{1}{k}$, $\theta_i = i$ (for $i = 1, 2, \dots, k$) and $k$ movies $s_j$ (for $j = 1, 2, \dots, k$) with $p_{s_j, i} = 1$ if $j = i$ and $p_{s_j, i} = 0$ otherwise.

The optimal list is $s_1 s_2 \dots s_k$, which covers each user type exactly before their patience expires, giving $OPT = 1$. Meanwhile, via induction on the iterations we see that the greedy algorithm can
choose the movies in reverse order, producing the list $s_k s_{k-1} \dots s_1$. Then only movies $s_k$ through $s_{k/2 + 1}$ will be able to interest their corresponding user type (for simplicity suppose $k$ is even); for movies $s_{k/2}$ through $s_1$, their corresponding user type will walk away before they are covered. So we have $$ALG = \sum_{i=1}^{k/2} \pi_i \cdot 0 + \sum_{i=k/2 + 1}^k \pi_i \cdot 1 = \frac{k}{2} \cdot \frac{1}{k} = \frac{1}{2}.$$
Again, $ALG/OPT = \frac{1}{2}$ exactly, establishing that the greedy approximation ratio of $2$ is tight.

\begin{thm}
There exist instances of ordered-submodular optimization problems on which the greedy algorithm achieves exactly $\frac{1}{2}$ of the optimal value. Thus, the 2-approximation performance bound is tight.
\end{thm}


%% file: sections/6_calibration.tex


\section{Application 2: Calibration in personalized recommendations}\label{sec:calibration}

We now consider the setting of \emph{personalized} recommendations, which generates a tailored list of recommendations for each individual user based on their historical preferences. Much research on personalized recommender systems has worked toward improving prediction accuracy (e.g., how many of the recommended items are indeed relevant to the user), but training solely toward accuracy metrics can actually be detrimental to the performance of the system. For instance, recommendation lists focused only on accuracy may suffer from a lack of diversity or novelty~\cite{McNee2006}. Another important metric in machine learning is \emph{calibration}, the degree to which the predicted proportions of the various classes align with the true proportions of the classes in the existing data. From the user's perspective, a recommendation list is calibrated if it closely reflects their various interests in appropriate proportions. This a desirable additional objective when optimizing the user experience; for instance, a user would likely want the system to preserve their minor interests, rather than entirely ``crowding them out'' in favor of major interests only.

Steck (2018) considers the problem of creating calibrated recommendations using the language of \emph{movies} as the items with which users interact, and \emph{genres} as the classes of items~\cite{Steck18}. Each user has a preference distribution over genres that can be inferred from their previous activity, and the goal is to recommend a list of movies whose genres reflect these preferences (possibly also incorporating a ``quality'' score for each movie, representing its general utility or relevance). In our work, we adopt Steck's formulation of distributions over genres, which we describe below.

Suppose that each movie $i$ has a distribution over genres $g$, given by $p(g|i)$. For a user $u$, we consider two induced distributions: one from the list of movies $\mathcal{H}$ that user $u$ has played in the past, and one from the list of movies $\mathcal{I}$ that the system recommends to user $u$:
\begin{itemize}
    \item $p(g|u)$, the distribution over genres $g$ played by user $u$ in the past: $$p(g|u) = \frac{\sum_{i\in \mathcal{H}} w_{u,i} \cdot p(g|i)}{\sum_{i\in \mathcal{H}} w_{u,i}},$$ where $w_{u,i}$ is the weight of movie $i$ (e.g., how recently it was played by user $u$),
    
    \item $q(g|u)$, the distribution over genres $g$ recommended to user $u$: $$q(g|u) = \frac{\sum_{i\in \mathcal{I}} w_{r(i)} \cdot p(g|i)}{\sum_{i\in \mathcal{I}} w_{r(i)}},$$ where $w_{r(i)}$ is the weight of movie $i$ due to its rank $r(i)$.
\end{itemize}

In general, Steck does not provide much guidance on how the weights are intended to be chosen and interpreted in the context of the greedy algorithm. For our purposes, we suppose the weights are weakly decreasing in rank (i.e., $w_a \ge w_b$ if $a < b$). We also suppose that the desired length of the recommendation list is a fixed constant $k$ (i.e., $|\mathcal{I}| = k$) and $\sum_{j=1}^k w_i = 1$. This assumption is without loss, even with the more typical cardinality constraint that the list may have length \emph{at most} $k$ --- we simply linearly consider each possible length $\ell \in [1,k]$, renormalize so that the first $\ell$ weights sum to $1$, and perform the optimization. We then take the maximally calibrated list over all $k$ length-optimal lists. 

The goal of the \emph{calibrated recommendations} problem is to choose $\mathcal{I}$ such that $q$ is ``close'' to $p$. To quantify this concept of closeness between distributions, we introduce the formalism of \emph{overlap measures}.

\subsection{Overlap measures}
For the discussion that follows, we restrict to finite discrete probability spaces $\Omega$ for simplicity, although the concepts can be generalized to continuous probability measures.

A common tool for quantitatively comparing distributions is statistical divergences, which measure the ``distance'' from one distribution to another. A divergence $D$ has the property that $D(p,q) \ge 0$ for any two distributions $p,q$, with equality attained if and only if $p=q$. This means that divergences cannot directly be used to measure calibration, which we think of as a non-negative metric that is uniquely \emph{maximized} when $p=q$. Instead, we define a new but closely related tool that we call \emph{overlap}, which exactly satisfies the desired properties. 

Since divergences have a number of well-studied properties and applications, it is useful to consider overlap measures derived from divergences. We note that~\cite{Steck18} does a version of this, modifying the KL divergence into a \emph{maximum marginal relevance} objective function. However, this proposed objective function may be either positive or negative (see Appendix~\ref{sec:bad_netflix} for an example), meaning that it cannot be used in our greedy algorithm---and in fact, the concept of approximation guarantees is not even well-specified for functions of variable sign. In contrast, our abstraction of overlap measures satisfies non-negativity for all pairs of distributions $p,q$.

Our definition is also more general in two important ways. First, we do not limit ourselves to the KL divergence, so that other divergences and distances with useful properties may be used (one such example is the Hellinger distance, $H(p,q) = \frac{1}{\sqrt{2}} ||\sqrt{p} - \sqrt{q} ||_2$, which forms a bounded metric and has a convenient geometric interpretation using Euclidean distance). Second, in our definition $q$ may be any \emph{subdistribution}, a term we use to denote a vector of probabilities summing to \emph{at most} $1$. This is useful because our greedy algorithm incrementally constructs $q$ from the $0$ vector by adding a new movie (weighted by its rank), so in each iteration we must compute the overlap between the true distribution $p$ and the partially constructed subdistribution $q$.

With these considerations in mind, we now proceed to define overlap measures.

\begin{defn}[Overlap measure]
An \textbf{overlap measure} $G$ is a function on pairs of distributions and subdistributions $(p,q)$ with the properties that
\begin{enumerate}[(i)]
    \item $G(p,q) \ge 0$ for all distributions $p$ and subdistributions $q$,
    \item For any fixed $p$, $G(p,q)$ is uniquely maximized at $q=p$.
\end{enumerate}
\end{defn}

Further, we observe that overlap measures can be constructed based on distance functions.

\begin{defn}[Distance-based overlap measure]
Let $d(p,q)$ be a bounded distance function on the space of distributions $p$ and subdistributions $q$ with the property that $d(p,q) \ge 0$ with $d(p,q) = 0$ if and only if $p=q$. Denote by $d^*$ the maximum value attained by $d$ over all pairs $(p,q)$.

Then, the \textbf{$d$-overlap measure} $G_d$ is defined as $$G_d(p,q) \coloneqq d^* - d(p,q).$$
\end{defn}

Many classical distances are originally defined on pairs of distributions $(p,q)$, but admit explicit functional forms that can be evaluated using the values of $p(x)$ and $q(x)$ for all $x \in \Omega$. This allows us to compute $d(p,q)$, and consequently $G_d(p,q)$, when $q$ is a subdistribution. Now, it is clear that $G_d$ indeed satisfies both properties of an overlap measure: property (i) follows from the definition of $d^*$, and property (ii) follows from the unique minimization of $d$ at $q=p$.

\subsection{Families of ordered-submodular overlap measures}
Any overlap measure produces a version of the calibrated recommendations problem (and the corresponding approximation problem), since for the user's target distribution $p(g|u)$, we seek to calibrate the recommended genre distribution $q(g|u)$ to equal $p(g|u)$, which maximizes the overlap $G(p,q)$. But to execute our greedy algorithm, we are interested in overlap measures that give rise to an ordered-submodular optimization problem in particular. As discussed in Section~\ref{sec:ordsub_defn}, it suffices to study conditions under which the resulting calibration heuristic exhibits diminishing marginal returns and monotonicity (with respect to filling in any position in the sequence that is currently the formal ``null'' with a new item). 

As before, we are also interested in divergence-based overlap measures. One of the main classes of divergences is the family of $f$-divergences, which are generated from functions $f(t)$ that are convex on $t \ge 0$ with $f(1) = 0$. Given such a function $f$, the $f$-divergence of $p$ from $q$ (alternatively, ``from $q$ to $p$'') is defined as $$D_f(p,q) \coloneqq \sum_{x\in \Omega} f\left( \frac{p(x)}{q(x)}\right) q(x).$$ We show that in general, $f$-divergences yield ordered-submodular problems to which we can apply our greedy algorithm.

Recall that $p$ is a given as fixed, while $q$ is constructed incrementally as $q(g|u) = \sum_{i\in \mathcal{I}} w_{r(i)} \cdot p(g|i)$. For any genre $g$, since the $p(g|i)$ are non-negative probabilities, adding a movie $i$ to the sequence always (weakly) increases $q(g|u)$. Then, in the execution of the greedy algorithm, adding $i$ later in the sequence adds onto a larger accumulated value of $q(g|u)$. If $g(y) \coloneqq f\left(\frac{c}{y}\right) y$ is convex for $y > 0$, then $D_f(p,q) = \sum_{y = q(x)} g(y)$ displays ``increasing marginal returns'' with respect to the sequence of movies, and $G_{D_f}(p,q)$ will display decreasing marginal returns.

Indeed, we can verify that we have 
\begin{align*}
    g'(y) &= -f'\left(\frac{c}{y}\right) \cdot \frac{c}{y^2} \cdot y + f\left(\frac{c}{y}\right) \cdot 1 = f\left(\frac{c}{y}\right) - f'\left(\frac{c}{y}\right) \frac{c}{y}, \\
    g''(y) &= -f''\left(\frac{c}{y}\right) \cdot \frac{c}{y^2} + f''\left(\frac{c}{y}\right) \cdot \frac{c}{y^2} \cdot \frac{c}{y} + f'\left(\frac{c}{y}\right) \cdot \frac{c}{y^2} \\
    &= f''\left(\frac{c}{y}\right) \frac{c}{y^3} \\
    &\ge 0,
\end{align*}
since $c,y, f''\left(\frac{c}{y}\right)  \ge 0$ (by definition of convexity of $f$). Consequently, any \emph{bounded} $f$-divergence results in a $D_f$-overlap measure with the diminishing returns property.

Further, we note that many $D_f$-overlap measures based on commonly used $f$-divergences satisfy the monotonicity property. As a concrete example, consider the squared Hellinger distance (obtained by choosing $f(t) = (\sqrt{t}-1)^2$ or $f(t) = 2(1-\sqrt{t})$), which is of the form $$H^2(p,q) = \frac{1}{2} \sum_{x\in \Omega} (\sqrt{p(x)} - \sqrt{q(x)})^2 = 1 - \sum_{x\in \Omega} \sqrt{p(x) \cdot q(x)}.$$ The resulting $H^2$-overlap measure is $$G_{H^2} (p,q) = \sum_{x\in \Omega} \sqrt{p(x) \cdot q(x)}.$$ Clearly, if subdistribution $q$ coordinate-wise dominates subdistribution $q'$, then $G_{H^2}(p,q) \ge G_{H^2}(p,q')$ for all $p$, which establishes the monotonicity property.

Taking these two desired properties together, we have demonstrated that our notion of overlap measures works well with $f$-divergences to create a general family of ordered-submodular calibration problems.

\begin{thm} \label{thm:overlap_fdiv}
Given any bounded $f$-divergence $D_f$ with maximum value $d^* = \max_{(p,q)} d(p,q)$, the corresponding $D_f$-overlap measure $$G_{D_f}(p,q) = d^* - D_f(p,q)$$ is ordered-submodular. Thus, the greedy algorithm provides a $2$-approximation for calibration heuristics using overlap measures of this form.
\end{thm}

Inspired by the squared Hellinger-based overlap measure, we also consider the construction of another general family of overlap measures of the form $$G(p,q) = \sum_{x\in \Omega} g_1(p(x)) \cdot g_2(q(x)),$$ for nonnegative functions $g_1$ and $g_2$. Here, $G$ is ordered-submodular with respect to the recommendation list as long as $g_2$ is any non-decreasing (corresponding to monotonicity) and concave (corresponding to diminishing returns) function. Given such a $g_2$, we fully specify the overlap measure by choosing $g_1$ such that $G$ is uniquely maximized when $q=p$.

That is, we consider the constrained maximization of $\sum_{i=1}^g g_1(p_i) \cdot g_2 (q_i)$, subject to $\sum_{i=1}^g q_i \le 1$. By placing a Lagrange multiplier of $\lambda$ on the constraint, we see that the maximum occurs when $$g_1 (p_i) \cdot g_2' (q_i) = \lambda$$ for all $i$. Since we would like this to be satisfied if $q_i = p_i$ for all $i$, and we can scale the overlap measure by a multiplicative constant without loss, it suffices to set $g_1$ identically to $\frac{1}{g_2'}$. This gives rise to a second family of ordered-submodular overlap measures that includes a wide range of general forms, and also remains easy to compute. 

\begin{thm} \label{thm:overlap_g1g2}
Given any nonnegative non-decreasing concave function $f$, the overlap measure $$G(p,q) = \sum_{x\in \Omega} \frac{f(q(x))}{f'(p(x))}$$ is ordered-submodular. Thus, the greedy algorithm provides a $2$-approximation for calibration heuristics using overlap measures of this form.
\end{thm}

As a concrete example, taking $f(x) = x^\alpha$ for $\alpha \in (0,1)$ gives $\frac{1}{f'(x)} = \alpha x^{1-\alpha}$, which produces the (scaled) overlap measure $$G(p,q) = \sum_{x\in \Omega} p(x)^{1-\alpha} q(x)^\alpha.$$ 
Observe that the natural special case of $\alpha = \frac{1}{2}$ gives $f(x) = \frac{1}{f'(x)} = \sqrt{x}$, providing an alternate construction that recovers the squared Hellinger-based overlap measure.

Finally, we note that until this point our discussion has focused only on calibration, but in practice the recommendations should also have high quality (utility, relevance, etc.). To address this, we can model the overall quality of a list as the sum of quality scores of the individual movies (as in~\cite{ashkan2015, Steck18}), and then optimize for a weighted sum of quality and calibration. As a straightforward sum of scores, the quality metric is modular (and hence ordered-submodular). Then by Lemma~\ref{lem:weighted_submodular} the combination of quality and calibration remains ordered-submodular, and our approximation results still hold.

\begin{cor}
The calibrated recommendations problem with combined quality and calibration metrics is ordered-submodular, and thus admits a 2-approximation via the greedy algorithm.
\end{cor}

\paragraph{Remark.} Since we have established in Theorems~\ref{thm:overlap_fdiv} and \ref{thm:overlap_g1g2} that the calibration heuristic can be formulated as ordered-submodular, our result in Proposition~\ref{prop:2approx} implies that the greedy algorithm gives us a 2-approximation to the calibrated recommendations problem. But unlike in Section~\ref{sec:calibration}, where we show that this factor of $2$ is tight for the coverage problem, we do not have a corresponding tight instance for calibration. \etcomment{I would have said "we do not have" (not just do not present} \ercomment{changed present -> have} In simulation, the greedy algorithm performs nearly optimally on most instances of the calibration problem; however, we leave stronger theoretical guarantees as an open problem.


%% file: sections/7_discussion.tex


\section{Discussion} \label{sec:discussion}

In this paper, we have presented a new definition of \emph{ordered submodularity}, which extends the traditional notion of set submodularity to a class of optimization problems in which the order of elements matters, because elements contribute differently based on their position in the sequence. In particular, our formalism models coverage and calibration of ranked lists, two standard problems in the design of content recommendation systems. We have also shown that greedy ordered-submodular maximization gives a $2$-approximation and that this bound is tight on simple instances of the coverage problem. This quantitative result establishes our framework as qualitatively distinct from previous formalisms of set and sequence submodularity, and thus our work has provided the first performance guarantee for approximate optimization of this type.

It is interesting to consider the greedy algorithm in the calibration problem and ask whether the factor of 2 is tight here too, or if the greedy algorithm always performs better in this specific context. Another potential direction for further investigation is parametrizing worst-case instances of the calibration problem, since we found the greedy solution to be very close to optimal across many randomly generated instances. More generally, we pose the natural open question: Does there exist a polynomial time approximation algorithm for ordered submodular maximization achieving a constant factor better than $2$? Or does the analogy to set submodularity continue to hold, in that the greedy algorithm provides the best approximation guarantee possible? Further understanding the approximability of this class of problems is a key next step in the continued development and application of our framework. 


%% file: sections/appendix.tex

\section{Appendix}

\subsection{Greedy algorithm on variants of the KL divergence} \label{sec:bad_netflix}
A natural hope might be to use the KL divergence as a calibration heuristic, as it is perhaps the most commonly used statistical divergence. Unfortunately, the KL divergence cannot be used directly because it is unbounded; our translation to the distance-based overlap measure is also not well-defined on the KL divergence for the same reason. In~\cite{Steck18} an alternative transformation is proposed, yielding the following calibration heuristic: 
$$f(\mathcal{I}) = \sum_g p(g|u) \log \sum_{i \in \mathcal{I}} w_{r(i)} \tilde{q} (g|i).$$
However, this objective function has inconsistent sign, depending on how the recommendation weights are chosen (and we note that Steck does not set any constraints on the weights), and consequently the greedy choice can be far from optimal. In fact, we show that the greedy solution can be negative, while the optimum is positive. So the KL divergence (and variants of it) are not conducive to multiplicative approximation guarantees for the calibration problem.

Suppose there are $4$ genres ($g_k$ for $k=1,2,3,4$), $2$ movies ($i_\ell$ for $\ell=1,2$), and $1$ user ($u$), and that we seek a recommendation list of length $2$ with weights $w_1 > w_2 = 1$. For simplicity of notation, we denote $p(g_k|u)$ as $p_k$. Suppose further that the movies have the following distributions over genres for some $\varepsilon \in (0,\frac{1}{3})$:

\begin{center}
\begin{tabular}{ p{4cm} p{4cm} }
$\tilde{q}(g_1 | i_1) = \frac{1}{2}(1-\varepsilon)$ & $\tilde{q}(g_1 | i_2) = \frac{1}{2}(1-\varepsilon)$ \\[6pt]
$\tilde{q}(g_2 | i_1) = \frac{1}{4}(1-\varepsilon)$ & $\tilde{q}(g_2 | i_2) = \frac{1}{2}(1-\varepsilon)$ \\[6pt]
$\tilde{q}(g_3 | i_1) = \frac{1}{4}(1-\varepsilon)$ & $\tilde{q}(g_3 | i_2) = \frac{\varepsilon}{2}$ \\[6pt]
$\tilde{q}(g_4 | i_1) = \varepsilon$ & $\tilde{q}(g_4 | i_2) = \frac{\varepsilon}{2}$ 
\end{tabular}
\end{center}

Finally, suppose the parameters are such that $$p_3 \log\left(\frac{1-\varepsilon}{2\varepsilon} \right) = (p_2-p_4) \log \left(\frac{2w_1 + 1}{w_1+2} \right).$$

Then, observe that 
\begin{align*}
    f(i_1 i_2) - f(i_2 i_1) &= p_2 \log \left( \frac{\frac{w_1}{4}(1-\varepsilon) + \frac{1}{2}(1-\varepsilon)}{\frac{w_1}{2}(1-\varepsilon) + \frac{1}{4}(1-\varepsilon)} \right) + p_3 \log\left( \frac{\frac{w_1}{4}(1-\varepsilon) + \frac{\varepsilon}{2}}{\frac{w_1\varepsilon}{2} + \frac{1}{4}(1-\varepsilon)} \right) + p_4 \log \left(\frac{w_1\varepsilon + \frac{\varepsilon}{2}}{\frac{w_1\varepsilon}{2} + \varepsilon} \right) \\
    &= p_2 \log\left(\frac{w_1 + 2}{2w_1 + 1} \right) + p_3 \left( \frac{w_1(1-\varepsilon) + 2\varepsilon}{2w_1\varepsilon + 1-\varepsilon} \right) + p_4 \log \left( \frac{2w_1 + 1}{w_1 + 2} \right) \\
    &= p_3 \left( \frac{w_1(1-\varepsilon) + 2\varepsilon}{2w_1\varepsilon + 1-\varepsilon} \right) + (p_4-p_2) \log \left( \frac{2w_1 + 1}{w_1 + 2} \right).
\end{align*}
We can verify that for $\varepsilon < \frac{1}{3}$, we have $\frac{w_1(1-\varepsilon) + 2\varepsilon}{2w_1\varepsilon + 1-\varepsilon} < \frac{1-\varepsilon}{2\varepsilon}$, thus 
\begin{align*}
    f(i_1 i_2) - f(i_2 i_1) &< p_3 \left( \frac{1-\varepsilon}{2\varepsilon}\right) + (p_4-p_2) \log \left( \frac{2w_1 + 1}{w_1 + 2} \right) = 0\\
    \implies f(i_1 i_2) &< f(i_2 i_1).
\end{align*}
That is, the optimal recommendation list ranks $i_2$ first, then $i_1$ second.

However, we also have
\begin{align*}
    f(i_1) - f(i_2) &= p_2 \log \left( \frac{\frac{w_1}{4}(1-\varepsilon)}{\frac{w_1}{2}(1-\varepsilon) } \right) + p_3 \log\left( \frac{\frac{w_1}{4}(1-\varepsilon)}{\frac{w_1\varepsilon}{2}} \right) + p_4 \log \left(\frac{w_1\varepsilon}{\frac{w_1\varepsilon}{2}} \right) \\
    &= p_2 \log\left(\frac{1}{2} \right) + p_3 \left( \frac{1-\varepsilon}{2\varepsilon} \right) + p_4 \log \left(2\right) \\
    &= p_3 \left( \frac{1-\varepsilon}{2\varepsilon} \right) + (p_4-p_2) \log \left(2 \right).
\end{align*}
Since $w_1 > 1$, we have $\frac{2w_1 + 1}{w_1+2} < 2$, thus
\begin{align*}
    f(i_1) - f(i_2) &> p_3 \left( \frac{1-\varepsilon}{2\varepsilon}\right) + (p_4-p_2) \log \left( \frac{2w_1 + 1}{w_1 + 2} \right) = 0\\
    \implies f(i_1) &> f(i_2).
\end{align*}
That is, the greedy algorithm will first choose $i_1$ instead of $i_2$, thereby constructing a suboptimal list.

Now, we compute $ALG = f(i_1i_2)$ and $OPT = f(i_2i_1)$ for the following set of parameters: $p_1 = 0.05, p_2 = 0.9, p_3=p_4 = 0.025$, $\varepsilon = 10^{-10}$, varying $w_1 > 1$. 
\begin{center}
\begin{tabular}{ c r r }
$w_1$ & $ALG$ & $OPT$ \\[3pt] \hline\\[-6pt]
1.1 & -0.823134 & -0.797737 \\[3pt]
1.5 & -0.691859 & -0.585156 \\[3pt]
2 & -0.549794 & -0.371873 \\[3pt]
3.5 & -0.201250 & 0.114023 \\[3pt]
5 & 0.0311358 & 0.386387 \\[3pt]
10 & 0.580034 & 1.01213 \\[3pt]
100 & 2.73099 & 3.20940 
\end{tabular}
\end{center}

We now observe that the function does not have consistent sign; $ALG$ and $OPT$ are negative for lower values of $w_1$ and positive for higher values of $w_1$. This is because the $\tilde{q}(g|i)$'s represent a probability distribution and are thus less than $1$, so when the weights are small we take the logarithm of a number less than $1$, so the function is negative; when the weights are sufficiently large, then the inner summand exceeds $1$ and the function becomes positive.

It is unclear how we should think about approximation when the value of a function is not always positive or negative --- for instance, the approximation ratio $ALG/OPT$ is meaningless, especially considering that $ALG$ and $OPT$ may have opposite signs (such as when $w_1 = 3.5$). So if the simple greedy algorithm is not always optimal, but we have no consistent way of comparing its performance with the optimal solution, then it becomes very difficult to understand the maximization (or approximate maximization) of this specific form of the calibration heuristic.

\subsection{Varying sequential dependencies in calibration}\label{sec:calibration_seqdep}

In Section~\ref{sec:relatedwork}, we described earlier formalisms of sequential submodularity that rely on postfix monotonicity and argued that many natural ordering problems, including the coverage objective function, are not postfix monotone. A different line of papers encodes sequences using DAGs and hypergraphs. Now, we show that this formalism also does not capture the rank-based sequential dependencies that we desire.

We present a simple instance of the calibration problem which hints at the potential intricacies of sequential dependencies. Suppose there are just $2$ genres ($g_1$ and $g_2$), $4$ movies ($i_1$, $i_2$, $i_3$, $i_4$), and $1$ user ($u$). Say that the target distribution is $p(g_1 | u) = p(g_2 | u) = 0.5$, and the weights of the recommended items are $w_1 = 0.5, w_2 = 0.3, w_3 = 0.2$. Suppose further that the movies have genre distributions as follows:

\begin{center}
\begin{tabular}{ p{3cm} p{3cm} }
$p(g_1 | i_1) = 0.4$, & $p(g_2 | i_1) = 0.6$ \\[6pt]
$p(g_1 | i_2) = 0.8$, & $p(g_2 | i_2) = 0.2$ \\[6pt]
$p(g_1 | i_3) = 1$, & $p(g_2 | i_3) = 0$ \\[6pt]
$p(g_1 | i_4) = 0$, & $p(g_2 | i_4) = 1$.
\end{tabular}
\end{center}

Our heuristic for measuring calibration is the overlap measure $G(p,q) = \sum_g \sqrt{p(g|u) \cdot q(g|u)}$. We now consider a few different recommended lists as input to the overlap measure:
\begin{align*}
    f(i_3 i_1 i_2) = G(p, (0.78,0.22)) \approx 0.956 \\
    f(i_3 i_2 i_1) = G(p, (0.82,0.18)) \approx 0.940 \\
    f(i_4 i_1 i_2) = G(p, (0.28,0.72)) \approx 0.974 \\
    f(i_4 i_2 i_1) = G(p, (0.32,0.68)) \approx 0.983 
\end{align*}

Here, we see that $f(i_3 i_1 i_2) > f(i_3 i_2 i_1)$, but $f(i_4 i_1 i_2) < f(i_4 i_2 i_1)$. So it is not always inherently better to rank $i_1$ before $i_2$ or $i_2$ before $i_1$; the optimal ordering is dependent on the context of the rest of the recommended list. Thus this very natural problem setting cannot be satisfactorily encoded by the DAG or hypergraph models of~\cite{Tschiatschek2017} and~\cite{Mitrovic18}, providing further motivation for our framework of ordered submodularity. 